\RequirePackage[orthodox]{nag}

\documentclass[a4paper,UKenglish]{lipics-v2019}
\usepackage{microtype}

\usepackage{amsthm,amssymb,amsmath,amsfonts}
\usepackage[ruled,vlined,linesnumbered]{algorithm2e}

\usepackage{xspace}
\newcommand{\reducefastest}[1]{\texttt{Reduce-Fastest\ensuremath{(#1)}}\xspace}
\newcommand{\reducemax}{\texttt{Reduce-Max}\xspace}
\newcommand{\merge}{merge\xspace}

\title{Cutting Bamboo Down to Size}

\author{Davide~Bilò}
{Department of Humanities and Social Sciences, University of Sassari, Italy.}
{davide.bilo@uniss.it}
{https://orcid.org/0000-0003-3169-4300}{This work was partially supported by the Research Grant FBS2016\_BILO, funded by "Fondazione di Sardegna" in 2016.}

\author{Luciano~Gualà}
{Department of Enterprise Engineering, University of Rome ``Tor Vergata'', Italy.}
{guala@mat.uniroma2.it}
{https://orcid.org/0000-0001-6976-5579}{}

\author{Stefano~Leucci}
{Department of Information Engineering, Computer Science and Mathematics, University of L'Aquila, Italy.}
{stefano.leucci@univaq.it}
{https://orcid.org/0000-0002-8848-7006}{}

\author{Guido~Proietti}
{Department of Information Engineering, Computer Science and Mathematics, University of L'Aquila, Italy. \\ Institute for System Analysis and Computer Science ``Antonio Ruberti'' (IASI~CNR), Italy.}
{guido.proietti@univaq.it}
{https://orcid.org/0000-0003-1009-5552}{}

\author{Giacomo~Scornavacca}{Department of Humanities and Social Sciences, University of Sassari, Italy.}{giacomo.scornavacca@graduate.univaq.it}{https://orcid.org/0000-0001-5921-0692}{This work was partially supported by Research Grant FBS2016\_BILO, funded by "Fondazione di Sardegna" in 2016.}

\authorrunning{Bilò et al.} 
\Copyright{Davide Bilò, Luciano Gualà, Stefano Leucci, Guido Proietti, and Giacomo Scornavacca}



\EventEditors{Martin Farach-Colton, Giuseppe Prencipe, and Ryuhei Uehara}
\EventNoEds{3}
\EventLongTitle{10th International Conference on Fun with Algorithms (FUN 2020)}
\EventShortTitle{FUN 2020}
\EventAcronym{FUN}
\EventYear{2020}
\EventDate{September 28--30, 2020}
\EventLocation{Favignana Island, Sicily, Italy}
\EventLogo{}
\SeriesVolume{157}
\ArticleNo{5}

\hideLIPIcs
\nolinenumbers

\begin{document}

\maketitle

\begin{abstract}
This paper studies the problem of programming a robotic panda gardener to keep a bamboo garden from obstructing the view of the lake by your house.

The garden consists of $n$ bamboo stalks with known daily growth rates and the gardener can cut at most one bamboo per day. As a computer scientist, you found out that this problem has already been formalized in [Gąsieniec et al., SOFSEM'17]  as the \emph{Bamboo Garden Trimming (BGT) problem}, where the goal is that of computing a perpetual schedule (i.e., the sequence of bamboos to cut) for the robotic gardener to follow in order to minimize the \emph{makespan}, i.e., the maximum height ever reached by a bamboo.

Two natural strategies are \reducemax and \reducefastest{x}. \reducemax trims the tallest bamboo of the day, while \reducefastest{x} trims the fastest growing bamboo among the ones that are taller than $x$. 
It is known that \reducemax and \reducefastest{x}  achieve a makespan of $O(\log n)$ and $4$ for the best choice of $x=2$, respectively. We prove the first constant upper bound of $9$ for \reducemax and improve the one for \reducefastest{x} to $\frac{3+\sqrt{5}}{2} < 2.62$ for $x=1+\frac{1}{\sqrt{5}}$.

Another critical aspect stems from the fact that your robotic gardener has a limited amount of processing power and memory. It is then important for the algorithm to be able to \emph{quickly} determine the next bamboo to cut while requiring at most linear space.
We formalize this aspect as the problem of designing a \emph{Trimming Oracle} data structure, and we provide three efficient Trimming Oracles implementing different perpetual schedules, including those produced by \reducemax and \reducefastest{$x$}.
\end{abstract}

\pagebreak

\section{Introduction}

You just bought a house by a lake. A bamboo garden grows outside the house and obstructs the beautiful view of the lake. 
To solve the problem, you also bought a robotic panda gardener which, once per day, can instantaneously trim a single bamboo. You have already measured the growth rate of every bamboo in the garden, and you are now faced with programming the gardener with a suitable schedule of bamboos to trim in order to keep the view as clear as possible. 

This problem is known as the {\em Bamboo Garden Trimming (BGT) Problem} \cite{DBLP:conf/sofsem/GasieniecKLLMR17} and can be formalized as follows: the garden contains $n$ bamboos $b_1,\ldots,b_n$, where bamboo $b_i$ has a known daily growth rate of $h_i > 0$, with $h_1 \geq \ldots \geq h_n$ and $\sum_{i=1}^{n} h_i = 1$. 
Initially, the height of each bamboo is $0$, and at the end of each day, the robotic gardener can trim at most one bamboo to instantaneously reset its height to zero. The height of bamboo $b_i$ at the end of day $d \geq 1$ and before the gardener decides which bamboo to trim is equal to $(d-d')h_i$, where $d' < d$ is the last day preceding $d$ in which $b_i$ was trimmed (if $b_i$ was never trimmed before day $d$, then $d'=0$). See Figure~\ref{fig:example} for an example.

The main task in BGT is to design a perpetual trimming schedule that keeps the tallest bamboo ever seen in the garden as short as possible. In the literature of scheduling problems, this maximum height is called \emph{makespan}. 

A simple observation shows that the makespan must be at least 1 for every instance. Indeed, for any $\epsilon>0$, a makespan of $1-\epsilon$ would imply that the daily amount of bamboo cut from the garden is at most $1-\epsilon$, while the overall daily growth rate of the garden is $1$. This is a contradiction. Furthermore, there are instances for which the makespan can be made arbitrarily close to 2. Consider, for example, two bamboos $b_1, b_2$ with daily growth rates $h_1=1-\epsilon$ and $h_2=\epsilon$, respectively. Clearly, when bamboo $b_2$ must be cut, the height of $b_1$ becomes at least $2-2\epsilon$. This implies that the best makespan one can hope for is 2. 

Two natural strategies are known for the BGT problem, namely \reducemax and \reducefastest{x}. The former consists of trimming the tallest bamboo at the end of every day, while the latter cuts the bamboo with fastest growth rate among those having a height of at least $x$. Experimental results show that \reducemax performs very well in practice as it seems to guarantee a makespan of 2~\cite{DBLP:conf/IEEEcit/AlshamraniKG15, DBLP:journals/algorithms/DEmidioSN19}. However, the best known upper bound to the makespan is $1+\mathcal{H}_{n-1} = \Theta(\log n)$, where $\mathcal{H}_{n-1}$ is the $(n-1)$-th harmonic number~\cite{DBLP:conf/ifipTCS/BodlaenderHKSWZ12}. Interestingly, this $\Theta(\log n)$ bound also holds for the adversarial setting in which at every day an adversary decides how to distribute the unit of growth among all the bamboos. In this adversarial case such upper bound can be shown to be tight, while understanding whether \reducemax achieves a constant makespan in the non-adversarial setting is a major open problem \cite{DBLP:conf/sofsem/GasieniecKLLMR17,DBLP:journals/algorithms/DEmidioSN19}.
On the other hand, in~\cite{DBLP:conf/sofsem/GasieniecKLLMR17} it is shown that \reducefastest{x} guarantees a makespan of $4$ for $x=2$. Furthermore, it is also conjectured that \reducefastest{1} guarantees a makespan of $2$~\cite{DBLP:journals/algorithms/DEmidioSN19}.

In~\cite{DBLP:conf/sofsem/GasieniecKLLMR17}, the authors also provide a different algorithm guaranteeing a makespan of 2. This is obtained by transforming the BGT problem instance into an instance of a related scheduling problem called \emph{Pinwheel Scheduling},
by suitably rounding the growth rates of the bamboos. 
Then, a perpetual schedule for the Pinwheel Scheduling instance is computed using existing algorithms~\cite{DBLP:journals/algorithmica/ChanC93, 48075}. It turns out that this approach has a problematic aspect since it is known that \emph{any} perpetual schedule for the Pinwheel Scheduling instance can have length $\Omega\left(\prod_{i=1}^{n} \frac{1}{h_i}\right)$ in the worst case.

\begin{figure}
    \centering
    \includegraphics{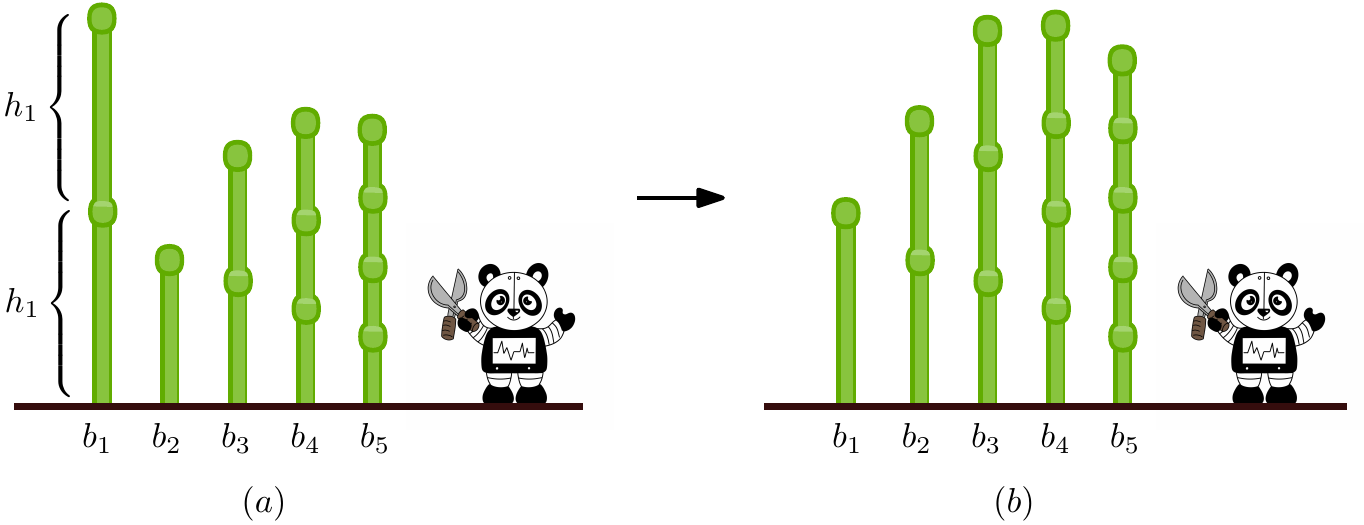}
    \caption{(a) The bamboo garden at the end of a day, just before the robotic gardener trims bamboo $b_1$. 
    (b) The bamboo garden at the end of the next day, before cutting a bamboo.}
    \label{fig:example} 
\end{figure}

The above observation gives rise to the following complexity issue: Can a perpetual schedule be efficiently implemented in general? Essentially, a solution consists of designing a \emph{trimming oracle}, namely a \emph{compact} data structure that is able to \emph{quickly} answer to the query ``What is the next bamboo to trim?''.

It is worth noticing that similar problems are discussed in \cite{DBLP:conf/sofsem/GasieniecKLLMR17}, where the authors ask for the design of trimming oracles that implement known BGT algorithms. For example, they explicitly leave open the problem of designing an oracle implementing \reducemax with query time of $o(n)$. 
\subparagraph{Our results.} Our contribution is twofold. In Section~\ref{sec:makespan}, we provide the following improved analyses of \reducemax and \reducefastest{x}:
\begin{itemize}
    \item We show that \reducemax achieves a makespan of at most $9$. This is the first constant upper bound for this strategy and shows a separation between the static and the adversarial setting for which the makespan is known to be $\Theta(\log n)$. 
    \item We show that, for any $x>1$, \reducefastest{x} guarantees a makespan of at most $\max\left\{x + \frac{x^2}{4(x-1)}, \frac{1}{2}  + x + \frac{x^2}{4(x-\frac{1}{2})}\right\}$. For the best choice of $x=1+\frac{1}{\sqrt{5}}$, this results in a makespan of $1 + \phi = \frac{3+ \sqrt{5}}{2} < 2.62$, where $\phi$ is the golden ratio. Notice also that for $x=2$ (the best choice of $x$ according to the analysis of \cite{DBLP:conf/sofsem/GasieniecKLLMR17}) we obtain an upper bound of $19/6$ which improves over the previously known upper bound of $4$. 
\end{itemize}

Then, in Section~\ref{sec:oracles}, we provide the following trimming oracles:
\begin{itemize}
    \item A trimming oracle implementing \reducemax 
    whose query time is $O(\log^2 n)$ in the worst-case or $O(\log n)$ amortized. The size of the oracle is $O(n)$ while the time needed to build it is $O(n \log n)$. This answers the open problem given in \cite{DBLP:conf/sofsem/GasieniecKLLMR17}.
    
    \item A trimming oracle implementing \reducefastest{x} with $O(\log n)$ worst-case query time. This oracle has linear size and can be built in $O(n \log n)$ time.
    
    \item A trimming oracle guaranteeing a makespan of 2. This oracle uses the rounding strategy from \cite{DBLP:conf/sofsem/GasieniecKLLMR17} 
    but it uses a different approach to compute a perpetual schedule. Our oracle answers queries in $O(\log n)$ amortized time, requires $O(n)$ space, and can be built in $O(n \log n )$ time.
    
    This result favorably compares with the existing oracles achieving makespan $2$ implicitly obtained  when the reduction of \cite{DBLP:conf/sofsem/GasieniecKLLMR17} is combined with the results in \cite{48075,DBLP:journals/algorithmica/ChanC93} for the Pinwheel Scheduling problem.
    Indeed, once the instance $G$ of BGT has been transformed into an instance $P$ of Pinwheel Scheduling, any oracle implementing a feasible schedule for $P$ is an oracle for $G$ with makespan $2$. In \cite{48075}, the authors show how to compute a schedule for $P$ of length $L=\Omega(\prod_{i=1}^{n} \frac{1}{h_i})$, which results in an oracle with exponential building time and constant query time.
    In \cite{DBLP:journals/algorithmica/ChanC93}, an oracle having query time of $O(1)$ is claimed, but attaining such a complexity requires the use of $\Theta(n)$ parallel processors and the ability to perform arithmetic operations modulo $L$ (whose binary representation may need $\Omega(n)$ bits) in constant time.
\end{itemize}

An interactive implementation of our Trimming Oracles described above is available at \url{https://www.isnphard.com/g/bamboo-garden-trimming/}.

\subparagraph{Other related work.} The BGT problem has been introduced in \cite{DBLP:conf/sofsem/GasieniecKLLMR17}. Besides the aforementioned results, this paper also provides an algorithm achieving a makespan better than 2 for a subclass of instances with \emph{balanced} growth rates; informally, an instance is said to be balanced if at least a constant fraction of the overall daily growth is due to bamboos $b_2, \dots, b_n$. The authors also introduce a generalization of the problem, named \emph{Continuous BGT}, where
each bamboo $b_i$ grows continuously at a rate of $h_i$ per unit of time and is located  in a point of a metric space. The gardener can instantaneously cut a bamboo that lies in its same location, but needs to move from one bamboo to the next at a constant speed.
Notice that this is a generalization of BGT problem since one can consider the trivial metric in which all distances are $1$ (and it is never convenient for the gardener to remain in the same location).

Another generalization of the BGT problem called \emph{cup game} can be equivalently formulated as follows: each day the gardener can reduce the height of a bamboo by up to $1+\epsilon$ units, for some constant parameter $\epsilon \ge 0$. If the growth rates can change each day and an adversary distributes the daily unit of growth among the bamboos, then a (tight) makespan of $O(\log n)$ can still be achieved. If the gardener's algorithm is randomized and the adversary is \emph{oblivious}, i.e., it is aware of gardener's algorithm but does not know the random bits or the previously trimmed bamboos, then the makespan is $O(m)$ with probability at least $1 - O(2^{-2^m})$, i.e., it is $O(\log \log n)$ with high probability \cite{KuszmaulSTOC2019}. The generalization of the cup game with multiple gardeners has been also addressed in \cite{KuszmaulSTOC2019, KuszmaulSODA2020a}. 

As we already mentioned, a problem closely related to BGT is the Pinwheel Scheduling problem that received a lot of attention in the literature \cite{DBLP:journals/tc/ChanC92,DBLP:journals/algorithmica/ChanC93,48075,DBLP:journals/tcs/HolteRTV92,DBLP:journals/algorithmica/LinL97,DBLP:journals/algorithmica/RomerR97}.

The BGT problem and its generalizations also appeared in a variety of other applications, ranging
from deamortization, to buffer management in network switches, to
quality of service in real-time scheduling (see, e.g., \cite{KuszmaulSODA2020b,GoldwasserSurvey,DBLP:conf/spaa/AdlerBFGGP03} and the references therein).

\section{New bounds on the makespan of known BGT algorithms}
\label{sec:makespan}

In this section we provide an improved analysis on the makespan guaranteed by the \reducefastest{x} strategy and the first analysis that upper bounds the makespan of \reducemax by a constant. In the rest of this section, we say that a bamboo $b_i$ is trimmed at day $d$ to specify that the schedule computed using the heuristic chooses $b_i$ as the bamboo that has to be trimmed at the end of day $d$. 

\subsection{The analysis for \reducemax}

Here we analyze the heuristic \reducemax, that consists in trimming the tallest bamboo at the end of each day (ties are broken arbitrarily). 

\begin{theorem}
    \label{thm:reduce_max}
	\reducemax guarantees a makespan of $9$.
\end{theorem}
\begin{proof}

We partition the bamboos into groups, that we call levels, according to their daily growth rates. More precisely, we say that bamboo $b_i$ is of \emph{level} $j \ge 1$ if $\frac{1}{2^j} \le h_i < \frac{1}{2^{j - 1}}$. Let $K$ be the level of bamboo $b_n$ and, for every $j$, with $1 \le j \le K$, let $L_j$ denote the set of all the bamboos of level $j$.

For every $1 \le j \le K$, let $\sigma(j)$ be the maximum height ever reached by any bamboo of level $k \ge j$, with $ \sigma(K+1) = 0$ by definition. In order to bound the makespan, it suffices to bound $\sigma(1)$. Rather than doing this directly, we will instead show that for $1 \le j \le K$, we have
\begin{equation}
\sigma(j) \le \max\Big\{3, \sigma(j + 1)\Big\} + 3 \sum_{k = 1}^j \frac{|L_k|}{2^j} .
\label{eq:level_comparison}
\end{equation}
Let $q \leq K$ be the level with lowest index such that $\sigma(q) \leq 3$ (if there is no such index, $q=K$). For any $j < q$ it holds $\max\Big\{3, \sigma(j + 1)\Big\} = \sigma(j + 1)$. As a consequence, the makespan is at most
\begin{equation}
\sigma(1) \le 3 + \sum_{j = 1}^q 3  \sum_{k = 1}^j \frac{|L_k|}{2^j}  \le 3 + 3 \sum_{j = 1}^{K} \sum_{k = 1}^j \frac{|L_k|}{2^j}.
\label{eq:span_bound}
\end{equation}
If bamboo $b_i$ is of level $s$, then the bamboo stalk contributes 
$\sum_{j = s}^{K} \frac{1}{2^j} <  \frac{2}{2^s} \leq 2 h_i$ to the sum in \eqref{eq:span_bound}. As $\sum_{i=1}^{n} h_i = 1$ by definition, it follows that the makespan is bounded by
\begin{equation*}
\sigma(1) \leq 3 + 3 \sum_{j = 1}^{K} \sum_{k = 1}^j \frac{|L_k|}{2^j} \leq 3 + 6 \sum_{i = 1}^n h_i = 9.
\label{eq:span_bound2}
\end{equation*}

We now complete the proof by proving \eqref{eq:level_comparison}, which compares $\sigma(j)$ and $\sigma(j + 1)$ for all $j$. Suppose that bamboo $b_i$ has level $j$, and that at the end of day $d_1$ bamboo $b_i$ achieves the maximum height ever reached by any bamboo of level $j$. Let $d_0 < d_1$ be the largest-numbered day prior to $d_1$ at the end of which either (a) a bamboo $b_\ell$ with level greater than $j$ was trimmed, or (b) a bamboo $b_\ell$ with height less than $3$ was trimmed. Because the \texttt{Reduce-Max} algorithm always trims the tallest bamboo, the height of $b_i$ at the end of day $d_0$ is at most the height of $b_\ell$ at the end of day $d_0$, right before $b_\ell$ is trimmed. It follows that the height of $b_i$ at the end of day $d_1$, right before $b_i$ is trimmed, is at most $h_i (d_1 - d_0)$ greater than the height of $b_{\ell}$ at the end of day $d_0$, right before $b_\ell$ is trimmed. Since the height of $b_\ell$ at the end of day $d_0$ is at most $\max\{3,\sigma(j + 1)\}$, it follows that
\begin{equation}
\sigma(j) \le \max\{3,\sigma(j + 1)\} + h_i(d_1 - d_0) < \max\{3,\sigma(j + 1)\} + \frac{2}{2^{j}} (d_1 - d_0),
\end{equation}
where in the last inequality we use the fact that $h_i < \frac{1}{2^{j - 1}}$. Now, in order to prove \eqref{eq:level_comparison}, it suffices to show that $d_1 - d_0 \le \frac{3}{2} \sum_{k = 1}^j |L_k|$. By the definition of $d_0$, at any day $t \in [d_0 + 1, d_1]$ a bamboo of height at least $3$ and with level equal or smaller than $j$ is trimmed. 
We call a cut at day $t \in [d_0 + 1, d_1]$ a \emph{repeated cut} if, at day $t$, a bamboo that was already trimmed at any day in $[d_0 + 1, t - 1]$ is trimmed again, and a \emph{first cut} otherwise. Note that each repeated cut trims a bamboo whose growth occurred entirely during days $[d_0 + 1, t - 1]$ and that the total growth of the forest in the interval interval $[d_0 + 1, d_1]$ is $d_1 - d_0$. It means that at most $\frac{1}{3}$ of the cuts at day $t \in [d_0 + 1, d_1]$ can be repeated cuts, since at the end of each of these days a bamboo of height at least $3$ is trimmed. On the other hand, the number of first cuts is bounded by the number of distinct bamboos with levels less or equal to $j$, i.e., by $\sum_{k = 1}^j |L_k|$. It follows that the number of days in the window $[d_0 + 1, d_1]$ satisfies $d_1 - d_0 \le \frac{1}{3} (d_1 - d_0) + \sum_{k = 1}^j |L_k|$, and thus $d_1 - d_0 \leq \frac{3}{2} \sum_{k = 1}^j |L_k|$ as desired. 
\end{proof}

\subsection{The analysis for \reducefastest{x}}

Here we provide an improved analysis of the makespan achieved by the \reducefastest{x} strategy. The heuristic \reducefastest{x} consists in trimming, at the end of each day, the bamboo with the fastest daily growth rate among those that have reached a height of at least $x$ (ties are broken in favour of the bamboo with the smallest index).

\begin{theorem}The makespan of \reducefastest{x}, for a constant $x$ such that $x>1$, is upper bounded by $\max\left\{x + \frac{x^2}{4(x-1)}, \frac{1}{2}  + x + \frac{x^2}{4(x-\frac{1}{2})}\right\}$.
\end{theorem}
\begin{proof}
Let $M$ be the makespan of \reducefastest{x} and let $b_i$ be one of the bamboos such that the maximum height reached by $b_i$ is exactly $M$. Let $[d_0, d_1]$ be an interval of days such that $b_i$ reaches the makespan in $d_1$ and $d_0$ is the last day in which $b_i$ was trimmed before $d_1$ ($d_0$ may also be equal to 0). Let $\delta$ the first day in $[d_0,d_1]$ such that the height of $b_i$ is at least $x$. For sake of simplicity we rename the interval $[\delta,d_1]$ as $[0,T]$, with $T = d_1-\delta$. Let $N$ be the number of distinct bamboos that are trimmed in $[0,T-1]$. 

We now give some definitions. Let the \emph{volume} $V$ of the garden be the overall growth of the bamboo in the days of the interval $[0, T-1]$. Since the garden grows by $\sum_{i=1}^n h_i = 1$ per day, we have $V=T$. Consider the cut of a bamboo $b_j$ on day $d \in [0, T-1]$. If $b_j$ was cut at least once in $[0, d-1]$ we say that the cut is a \emph{repeated cut} otherwise we will say that the cut is a \emph{first cut}.
The act of cutting bamboo $b_j$ on a day $d \in [0, T-1]$ with a repeated cut \emph{removes} an amount of volume that is equal to $(d-d')h_j$, where $d'$ is the last day of $[0, d-1]$ in which $b_j$ has been cut, if this is a repeated cut, and $d'=0$ if this is a first cut. Finally, the \emph{leftover volume} of a bamboo $b_j$ is the overall growth of $b_j$ that happened during interval $[0, T-1]$ and has not been cut by the end of day $T-1$.

We will now bound the amount $V'$ of volume $V$ that is removed by repeated cuts in the interval $[0,T-1]$.
Notice that, for each bamboo $b_j$ that is cut in the interval $[0,T-1]$, it holds that $h_j \geq h_i$. If $b_j$ is cut for its first time at day $d$ (among the days in $[0,T-1]$), then the removed volume will be at least $(d+1) h_j \ge (d+1) h_i$. Therefore, after all the $N$ bamboos of the interval $[0,T-1]$ have been cut at least once, the amount of volume removed by first cuts will be at least $\sum_{j=i}^N j h_i = \frac{N(N+1)}{2} \cdot  h_i$, since at most one bamboo is cut per day.
Moreover, if $b_j$ is cut for its last time at day $T-1-d$ (among the days in $[0,T-1]$), $b_j$ will have a height of $d h_i$ at the end of day $T-1$. Finally, bamboo $h_i$ is never cut in the interval $[0,T-1]$ and hence during the interval $[0, T-1]$ it grows by exactly $T h_i$.
This means that the overall leftover volume will be at least $\sum_{j=1}^N (j-1) h_i + T h_i = \frac{N(N-1)}{2}\cdot h_i + T h_i$.

We can then write
\[
V' \le V - \left( \frac{N(N+1)}{2} + \frac{N(N-1)}{2} \right)  \cdot  h_i - T h_i = V - N^2 h_i - T h_i = T(1-h_i) - N^2 h_i,
\]
where the last equality follows from $V=T$.

Since in $[0,T-1]$ the bamboo $b_i$ has height at least $x$, each repeated cut removes at least $x$ units of volume from $V'$. Therefore, the number $N'$ of repeated cuts is at most $\frac{V'}{x} \leq \left( T(1-h_i) - N^2 h_i \right)/x$. We now use this upper bound on $N'$ to derive an upper bound to the time $T$:
\[
T = N + N' \leq N + \frac{T(1-h_i) - N^2 h_i}{x}.
\]

For $T'(N)=(Nx  - N^2 h_i) / (h_i +x-1)$, the above formula
implies $T \le T'(N)$. If we fix $h_i$ and $x$, $T'(N)$ is a concave downward parabola that attains its maximum in its vertex at $N = x/2h_i$. Thus:
\[
T \le T'(x/2h_i) \le 
\frac{\frac{x^2}{2 h_i} - \frac{x^2}{4 h_i}}{h_i +x-1} 
= \frac{x^2}{4 h_i(h_i +x-1)}.
\]
Using this upper bound to $T$ we now bound the overall growth of the bamboo $b_i$, i.e., the makespan $M$. At day $d=0$, $b_i$ has height at most $x+h_i$ by our choice of $\delta$, and in the next $T$ days it grows by $T h_i$. Hence:
\begin{equation}\label{eq:makespan}
M \leq x + h_i + T h_i  < x + h_i + \frac{x^2}{4(h_i + x-1)}. 
\end{equation}
Let $M'(h_i)=x + h_i + \frac{x^2}{4(h_i + x-1)}$.
The derivative w.r.t.\ $h_i$ of the above formula is 
\begin{align*}
\frac{\partial M'}{\partial h_i}
&=1 - x^2/4(h_i + x - 1)^2 
= \frac{4(h_i + x - 1)^2  -  x^2}{4(h_i + x - 1)^2 }
= \frac{(x + 2 h_i -2) (3 x + 2 h_i -2  )}{4 (h_i + x - 1)^2}.
\end{align*}
The denominator is always positive, and the numerator is a concave upward parabola having its two roots at $h_i = 1 - 3x/2$ and at $h_i = 1 - x/2$.  Let us briefly restrict ourselves to the case $h_i \le \frac{1}{2}$ and notice that, since $x>1$, the first root is always negative, while the second root is always smaller than $\frac{1}{2}$. It follows that the maximum of $M'(h_i)$ is attained either at $h_i=0$ or at $h_i = \frac{1}{2}$.
Substituting in  Equation~\ref{eq:makespan} we get:
\[
    M \le \max\left\{ x + \frac{x^2}{4(x-1)},  x + \frac{1}{2} + \frac{x^2}{4(x-\frac{1}{2})} \right\}
\]

As far as the case $h_i > \frac{1}{2}$ is concerned, notice that it implies $i=1$ (since if $i \ge 2$ we would have the contradiction $\sum_{i=1}^n h_i > \frac{1}{2} \cdot i= 1$) and hence bamboo $b_1$ is trimmed as soon as its height reaches at least $x$. The makespan $M$ must then be less than $x + h_1 < x + 1$, which is always smaller than $x+\frac{1}{2} + \frac{x^2}{4(x-\frac{1}{2})} > x + \frac{1}{2} + \frac{1}{2}$. 
\end{proof}
\begin{corollary}\label{cor:makespan}
The makespan of \reducefastest{2} is at most $19/6$ and the makespan of \reducefastest{1+\frac{1}{\sqrt{5}}} is at most $1+\phi < 2.62$, where $\phi$ is the golden ratio.
\end{corollary}

\section{Trimming oracles}
\label{sec:oracles}

This section is devoted to the design of trimming oracles. More precisely, we first design two trimming oracles that implement \reducefastest{x} and \reducemax, respectively. The trimming oracle that implements \reducefastest{x} has a $O(\log n)$ worst-case query time, uses linear size and can be built in $O(n \log n)$ time. The trimming oracle that implements \reducemax has a $O(\log^2 n)$ worst-case query time or a $O(\log n)$ amortized query time, uses linear space, and can be built in $O(n \log n)$ time. We conclude this section by designing a novel trimming oracle  that guarantees a makespan of $2$ and has a $O(\log n)$ amortized query time. The oracle uses linear size and can be built in $O(n\log n)$ time. For technical convenience, in this section we index days starting from $0$, so that at the end of day $0$ the gardener can already trim the first bamboo.

An interactive implementation of the Trimming Oracles described in this section is available at \url{https://www.isnphard.com/g/bamboo-garden-trimming/}. 

\subsection{A Trimming Oracle implementing \reducefastest{x}}
\label{sec:oracle_reducefastest}

We now describe our trimming oracle implementing \reducefastest{x}.
The idea is to keep track, for each bamboo $b_i$, of the next day $d_i$ at which $b_i$ will be at least as tall as $x$.
When a query at a generic day $D$ is performed, we will then return the bamboo $b_i$ with minimum index $i$ among the ones for which $d_i \ge D$.

To this aim we will make use of a \emph{priority search tree} \cite{DBLP:journals/siamcomp/McCreight85} data structure $T$ to dynamically maintain a collection $P = \{ (x_1, y_1), (x_2, y_2), \dots \}$ of 2D points with distinct $y$ coordinates in $\{1, \dots, n\}$ under insertions and deletions while supporting the following queries:
\begin{description}
    \item[MinYInXRange($T, x_0$):] report the minimum $y$-coordinate among those of the points $(x_i, y_i)$ for which $x_i \le x_0$, if any. 
    \item[GetX($T, y$):] report the $x$-coordinate $x_i$ of the (at most one) point $(x_i, y_i)$ for which $y_i = y$, if any.
\end{description}

All of the above operations on $T$ require time $O(\log |P|)$,  as long as all coordinates and query parameters fit in $O(1)$ words of memory.\footnote{While this query is not described in \cite{DBLP:journals/siamcomp/McCreight85}, it can be  easily implemented in $O(\log |P|$) time using a dictionary and the fact that $y$-coordinates are distinct.}

In our case, the points $(x_i, y_i)$  will be the pairs $(d_i, i)$ for $i=1,\dots,n$. In such a way, a MinYInXRange query with $x_0=D$ will return exactly the index $i$ of the bamboo $b_i$ to be cut at the end of day $D$, if any. 
After cutting $b_i$, we \emph{update} $T$ to account for the new day at which the height $b_i$ will be at least $x$, i.e., we replace the old point $(d_i, i)$ with $(D + \lceil x / h_i \rceil, i)$.
Unfortunately, since the trimming oracle is ought to be used perpetually, (the representations of) both $d_i$ and $D$ will eventually require more than a constant number of memory words.

We solve this problem by dividing the days into contiguous intervals $I_0, I_1, \dots$ of $n$ days each, where $I_j = [ nj, nj+1, \dots, n(j+1)-1]$, and by using two priority search trees $T_1$ and $T_2$ that are associated with the current and the next interval, respectively. 
This allows us to measure days from the start of the current interval $I_j$, i.e., if $D=nj+\delta \in I_j$, then we only need to keep track of $\delta \in [0, \dots, n-1]$.
In place of $(d_i, i)$, we store the point $(\delta_i, i)$ in $T_1$, where $\delta_i = d_i - nj$. In this way, the previous query with $x_0 = D$ will now correspond to a query with $x_0 = \delta$.

Finally, we also ensure that at the end of the generic day $D=nj+\delta$, $T_2$ contains the point $(\delta'_i, i)$ for each $d_i = n(j+1) + \delta'$ and $i =1, \dots, \delta+1$. This allows us to swap $T_2$ for $T_1$ when interval $I_{j}$ ends.

\begin{figure}
    \centering
    \includegraphics{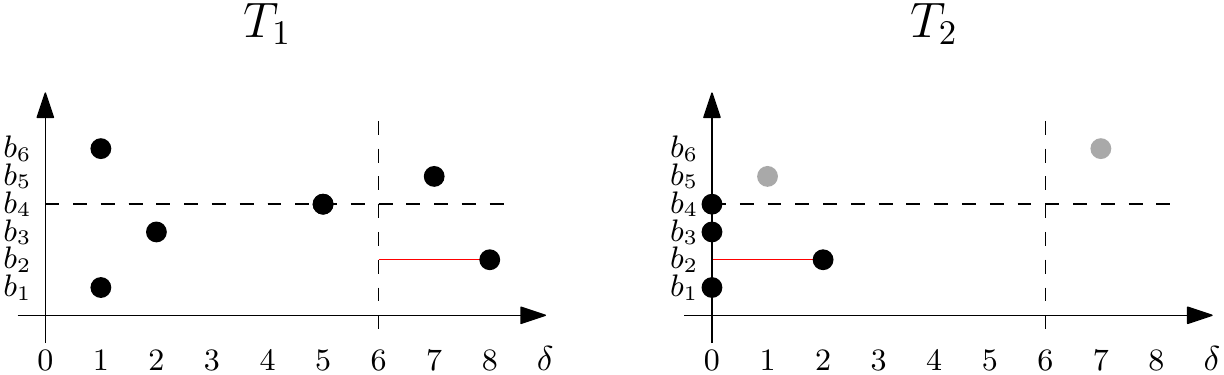}
    \caption{An example of the points contained in the priority search trees $T_1$ and $T_2$ for an instance with $6$ bamboos at the end of day $\delta=4$ of a generic interval $I_j$. We labeled the $y$-coordinate $i$ with $b_i$ since the unique point $(d_i, i)$ having $y$-coordinate $i$ represents the day at which $b_i$ reached/will reach a height of at least $x$. Notice that the points corresponding to bamboos $b_1$, $b_2$, $b_3$, and $b_4$ are already updated in $T_2$, while $b_5$ and $b_6$ (shown in gray) will be updated by the days $\delta=5$ and $\delta=6$, respectively. At the end of day $\delta=6$, all the points in $T_2$ are updated and $T_1$ can be safely swapped with $T_2$.}
    \label{fig:reduce_fastest_oracle}
\end{figure}

Since bamboo $b_i$ reaches height $x$ exactly  $\lceil x/h_i \rceil$ days after being cut, it follows that the largest $x$-coordinate ever stored in $T_1$ or $T_2$ is at most $n + x/h_n$ and we can then support MinYInXRange queries in $O(\log(n))$ time (where we are assuming that $x$, $h_n$ and thus $x/h_n$ fit in a constant number of memory words).

The pseudocode of our trimming oracle is as given in Algorithm~\ref{alg:reduce_fastest_oracle}. The procedure Query() is intended to be run every day. Consider a generic day $\delta$ of the current interval $I_j$. At this time, $T_1$ correctly encodes all the days at which the bamboos reached, or will reach, height at least $x$ when measured from the starting day of the current interval (i.e., from day $nj$), and after all the cuts of the previous days have already been performed.\footnote{Actually, if a bamboo $b_i$ reached height $x$ before the beginning of the considered interval, we will store the point $(0, i)$ in place of $(\delta_i, i)$ with $\delta_i < 0$. This still encodes the fact that it is possible to trim $b_i$ from the very fist day of the interval and prevents $\delta_i$ from becoming arbitrarily small.}
The same information concerning bamboos $b_1, \dots, b_{\delta}$ is replicated in $T_2$ with respect to the starting time of the next interval (i.e., $(n+1)j$). The procedure Query() accomplishes two tasks: (1) it computes the bamboo $b_i$ to cut at the end of day $\delta$ of the current interval (if any) and it updates the data structures $T_1$ and $T_2$ to account for the new height of $b_i$; (2) it updates the information concerning $b_{\delta+1}$ in $T_2$. See Figure~\ref{fig:reduce_fastest_oracle} for an example.

\begin{algorithm}[t]
\small
\caption{Trimming Oracle for \reducefastest{$x$}}
\label{alg:reduce_fastest_oracle}

\SetKwProg{Function}{Function}{:}{}
\SetKwFunction{Query}{Query}
\SetKwFunction{Update}{Update}
\SetKwFunction{Build}{Build}

\SetKwFunction{PriorityQuery}{MinYInXRange}
\SetKwFunction{PriorityInsert}{Insert}
\SetKwFunction{PriorityDelete}{Delete}
\SetKwFunction{PriorityGetX}{GetX}

\Function{\Build{}}
{
    $\delta \gets 0$\;
    $T_1, T_2 \gets $ Pointers to two empty priority search trees\;
    
    \BlankLine
    $h_1, \dots, h_n \gets$ Sort the growth rates of the $n$ bamboo in nonincreasing order\;
    \For{$i=1\,\dots,n$}
    {
        \PriorityInsert{$T_1, (\lceil x/h_i \rceil - 1, i)$}
    }
}

\BlankLine

\Function{\Update{$T, \delta_i, i$}}
{
    $\delta'_i \gets \PriorityGetX(T, i)$\;
    \lIf{$\delta'_i$ exists}{\PriorityDelete{$(\delta'_i, i)$}}
    \PriorityInsert{$T,  (\max\{0, \delta_i\}, i)$}\;
}

\BlankLine

\Function{\Query{}}
{
    \tcp{Cut fastest bamboo $b_i$ that reached height $x$ by day $\delta$}
    $i \gets \PriorityQuery(T_1, \delta)$\;
    \If{$i$ exists}
    {
        \Update{$T_1$, $\delta + \lceil x/h_i \rceil$, $i$} \;
        \Update{$T_2$, $\delta + \lceil x/h_i \rceil -n$, $i$} \;
    }
    
    \BlankLine
    
    \tcp{Make sure that bamboo $b_{\delta+1}$ is updated in $T_2$}
    $\delta_{\delta+1} \gets \PriorityGetX(T_1, \delta+1)$\;
    \Update{$T_2$, $\delta_{\delta+1} - n$, $\delta + 1$}

    \BlankLine
    \tcp{Move to the next day and possibly to the next interval}
    $\delta \gets (\delta + 1) \bmod n$\;
    \lIf{$\delta=0$}{Swap $T_1$ and $T_2$}

    \BlankLine

    \leIf{$i$ exists}{\Return ``Trim bamboo $b_i$''}{\Return ``Do nothing''}
}
\end{algorithm}

The following theorem summarizes the performances of our trimming oracle.
\begin{theorem}
There is a Trimming Oracle implementing \reducefastest{x} that uses $O(n)$ space, can be built in $O(n \log n)$ time, and can report the next bamboo to trim in $O(\log n)$ worst-case time.
\end{theorem}

\subsection{A Trimming Oracle implementing \reducemax}

The idea is to maintain collection $L$ of $n$ lines $\ell_1, \dots, \ell_n$ in which $\ell_i(d) = h_i d + c_i$ is associated with bamboo $b_i$ and represents its height at the end of day $d$. Initially $c_i = h_i$.

Determining the bamboo $b_i$ to trim at a generic day $d$ then corresponds to finding the index $i$ that maximizes $ \ell_i(d)$. After bamboo $b_i$, previously of height $H$, has been cut, $\ell_i$ needs to be updated to reflect the fact that $b_i$ has height $0$ at time $d$, which corresponds to decreasing $c_i$ by $H$.

\begin{figure}
    \centering
    \includegraphics{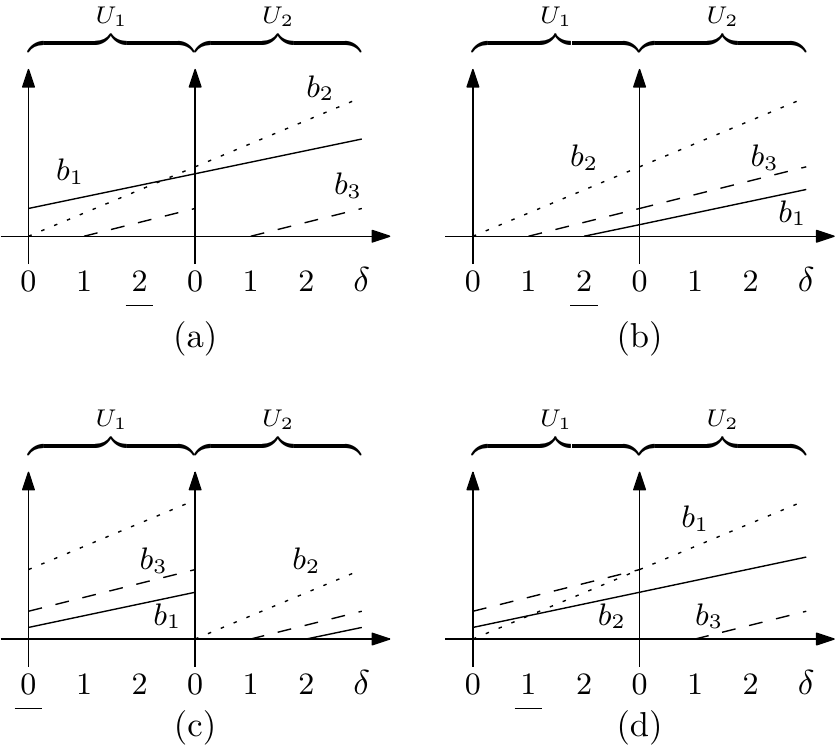}
    \caption{An example of the points contained in $U_1$ and $U_2$ for an instance with $3$ bamboos, at the beginning of day $2$ of a generic interval $I_j$ (a), at the end of day $2$ of $I_j$ but before moving to $I_{j+1}$ (b), at beginning of day $0$ of $I_{j+1}$ (c), and at the beginning of day $1$ of $I_{j+1}$ (d).}
    \label{fig:reduce_max_oracle}
\end{figure}

The \emph{upper envelope} $\mathcal{U}_L$ of $L$ is a function defined as $\mathcal{U}_L(d) = \max_{\ell \in L}{\ell(d)}$. We make use of an \emph{upper envelope data structure} $U$ that is able to maintain $L$ under insertions, deletions and lookups of named lines and supports the following query operation:
\begin{description}
\item[Upper($U, d$)] return a line $\ell  \in L$ for which $\ell(d) = \mathcal{U}_L(d)$.
\end{description}

Unfortunately, the trivial implementation of the trimming oracle suggested by the above description encounters similar problems as the ones discussed in Section~\ref{sec:oracle_reducefastest} for \reducefastest{x}:  the current day $d$ and the coefficients $c_i$ will grow indefinitely, thus affecting the computational complexity.

Once again, we solve this problem by using two copies $U_1$, $U_2$ of the previous \emph{upper envelope data structure} and by dividing the days into intervals $I_1, I_2,\dots$ with $I_j = [ nj, nj+1, \dots, n(j+1)-1]$. 
At the beginning of the current day $D = nj + \delta \in I_j$, $U_1$ will contain all lines $\ell_1, \dots, \ell_n$ and the value of each $\ell_i(\delta)$ will be exactly the height of $b_i$.
Moreover, at the end of day $D$ (i.e., after the highest bamboo of day $D$ has been trimmed), $U_2$ will contain a line $\ell'_{i}$ for each $i \le \delta+1$ such that $\ell'_i(\delta')$ with $\delta' \in [0, n-1]$ is exactly the height reached by $b_i$ on day $n(j+1)+\delta'$ if it is not trimmed on days $nj+\delta+1, \dots, n(j+1)+\delta'-1$. 
This means that at the end of day $nj + (n-1)$, $U_2$ correctly describes the heights of all bamboos in the next interval $I_{j+1}$ as a function of $\delta'$, and we can safely swap $U_1$ with $U_2$. See Figure~\ref{fig:reduce_max_oracle} for an example.

The pseudocode of our trimming oracle is given in Algorithm~\ref{alg:trimming_oracle_reducemax}. A technicality concerns the initial construction of the set of lines in $U_1$.
Notice that this is not handled by the Build() function, but we iteratively add $\ell_1, \dots, \ell_n$ during the first $n$ calls to Query() (i.e., during the days of interval $I_0$).
We can safely do this since the \reducemax strategy ensures that at time $D \in I_0$ only  bamboos in $\{b_1, \dots, b_{D+1}\}$ can conceivably be trimmed. This is handled by the test of line \ref{ln:test_start}, which is only true for $D \in I_0$ and will impact our amortized bounds, as noted below.

\begin{algorithm}[t]
\small 
\caption{Trimming Oracle for \reducemax}
\label{alg:trimming_oracle_reducemax}

\SetKwProg{Function}{Function}{:}{}
\SetKwFunction{Query}{Query}
\SetKwFunction{Update}{Update}
\SetKwFunction{Build}{Build}

\SetKwFunction{Upper}{Upper}
\SetKwFunction{Insert}{Insert}
\SetKwFunction{Delete}{Delete}

\Function{\Build{}}
{
    $\delta \gets 0$\;
    $T_1, T_2 \gets $ Pointers to two empty upper envelope data structures\;
    $h_1, \dots, h_n \gets$ Sort the growth rates of the $n$ bamboo in nonincreasing order\;
}

\Function{\Update{$U, i, c$}}
{
    \Delete{$U, \ell_i$}\;
    \Insert{$U, \ell_i(d) = h_i d + c$}\;
}

\BlankLine

\Function{\Query{}}
{
    \tcp{Ensure that the line $\ell_{\delta+1}$ corresponding to bamboo $b_{\delta+1}$ is in $U_1$}
    \If{there is no line named $\ell_{\delta+1}$ in $U_1$ \label{ln:test_start}}
    {
        \Insert{$U_1, \ell_{\delta+1}(d) = h_{\delta+1} d + h_{\delta+1}$}\label{ln:test_end}\;
    }
    \BlankLine
    
    \tcp{Cut highest bamboo $b_i$ at day $\delta$}
    $\ell_i(d) = h_id + c_i \gets$ \Upper{$\delta$}\;
    \Update{$U_1, i,  -\delta h_i$}\;
    \Update{$U_2, i, (n-\delta) h_i$}\;
    
    \BlankLine
    \tcp{Ensure that the line $\ell_{\delta+1}$ corresponding to bamboo $b_{\delta+1}$ is updated in $U_2$}
    Let $\ell_{\delta +1}(d) = h_{\delta+1}d+c_{\delta+1}$ be the line named $\ell_{\delta+1}$ in $U_1$\;
    \Update{$U_2, \delta+1, nh_{\delta+1}+c_{\delta+1}$}\;
    
    \BlankLine
    \tcp{Move to the next day and possibly to the next interval}
    $\delta \gets (\delta + 1) \bmod n$\;
    \lIf{$\delta=0$}{Swap $U_1$ and $U_2$}
    
    \BlankLine
    \Return ``Trim bamboo $b_i$''\;
}
\end{algorithm}

The performances of our trimming oracle depend on the specific implementation of the upper envelope data structure use. In \cite{DBLP:journals/jcss/OvermarsL81}, such a data structure guaranteeing a worst-case time of $O(\log^2 n)$ per operation is given, while a better amortized bound of $O(\log n)$ per operation was obtained in \cite{DBLP:conf/focs/BrodalJ02}.\footnote{Actually, the authors of \cite{DBLP:journals/jcss/OvermarsL81} and \cite{DBLP:conf/focs/BrodalJ02} design a dynamic data structure to maintain the convex hull of a set
of points in the plane. As explained in \cite{DBLP:conf/focs/BrodalJ02}, point-line duality can be used to convert such
a structure into one maintaining the upper envelope of a set of linear functions.} 
Moreover, from Theorem~\ref{thm:reduce_max} we know that the makespan of \reducemax{} is at most constant, implying that the maximum absolute value of a generic coefficient $c_i$ is at most $O(n h_i) = O(n)$.

The following theorem summarizes the time complexity of our trimming oracle.\footnote{Due to lines \ref{ln:test_start} and \ref{ln:test_end},
the complexity of a query operation is only amortized over the running time of previous queries.}

\begin{theorem}
There is a Trimming Oracle implementing \reducemax that uses $O(n)$ space, can be built in $O(n \log n)$ time, and can report the next bamboo to trim in $O(\log^2 n)$ worst-case time, or $O(\log n)$ amortized time.
\end{theorem}

\subsection{A Trimming Oracle achieving makespan 2}

We now design a Trimming Oracle implementing a perpetual schedule that achieves a makespan of at most $2$.

We start by rounding the rates $h_1, \dots, h_n$ down to the previous power of $\frac{1}{2}$ as in \cite{DBLP:conf/sofsem/GasieniecKLLMR17}, i.e., we set $h'_i = 2^{\lfloor \log_2 h_i \rfloor}$.
We will then provide a perpetual schedule for the rounded instance achieving makespan at most $1$ w.r.t.\ the new rates $h'_1, \dots, h'_2$. Since $h_i \le 2 h'_i$, this immediately results in a schedule having makespan at most $2$ in the original instance.

Henceforth we assume the input instance is already such that each $h_i$ is a power of $\frac{1}{2}$. Moreover, we will also assume that $\sum_{i=1}^n h_i = 1$. Indeed, if $\sum_{i=1}^n h_i <1$ then we can artificially increase some of the growth rates to meet this condition. Clearly, any schedule achieving makespan of most $1$ for the  transformed instance, also achieves makespan at most $1$ in the non-transformed instance.

We transform the instance as follows: we iteratively consider the bamboos in nonincreasing order of rates; when $b_i$ is considered we update $h_i$ to $2^{\left\lfloor \log_2\left( 1-\sum_{j \neq i} h_j\right)\right\rfloor}$, i.e., to the highest rate that is a power of $\frac{1}{2}$ and still ensures that the sum of the growth rates is at most $1$. One can easily check that the above procedure yields an instance for which $\sum_{i=1}^n h_i = 1$, as otherwise $\sum_{i=1}^n h_i < 1$ and $1 - \sum_{i=1}^{n} h_i \ge h_n$, which is a contradiction since $h_n$ would have been increased to $2h_n$. This requires $O(n \log n)$ time.

In the rest of this section, we will design  Trimming Oracles achieving a makespan of at most $1$ for instances where all $h_i$s are powers of $\frac{1}{2}$ and $\sum_{i=1}^n h_i = 1$. 

\subsubsection*{A Trimming Oracle for regular instances}

Let us start by considering an even smaller subset of the former instances, namely the ones in which $b_i$ has a growth rate of $h_i = 2^{-i}$, for $i=1,\dots,n-1$, and $h_n = h_{n-1} = 2^{-n+1}$. For the sake of brevity we say that these instances are \emph{regular}.\footnote{Notice that, in any regular instance, the grow rates of the bamboos are completely specified by the number $n$.}

It turns out that a schedule for regular instances can be easily obtained by exploiting a connection between the index $i$ of bamboo $b_i$ to be cut at a generic day $D$ and the position of the least significant $0$ in the last $n-1$ bits in the binary representation of $D$.

The schedule is as follows: if the last $0$ in the binary representation of $D$ appears in the $i$-th least significant bit, with $i < n$, then $b_i$ is to be cut at the end of day $D$. Otherwise, if the $n-1$ least significant bits of $D$ are all $1$, bamboo $b_n$ is cut at day $D$.

\begin{figure}[t]
    \centering
    \includegraphics{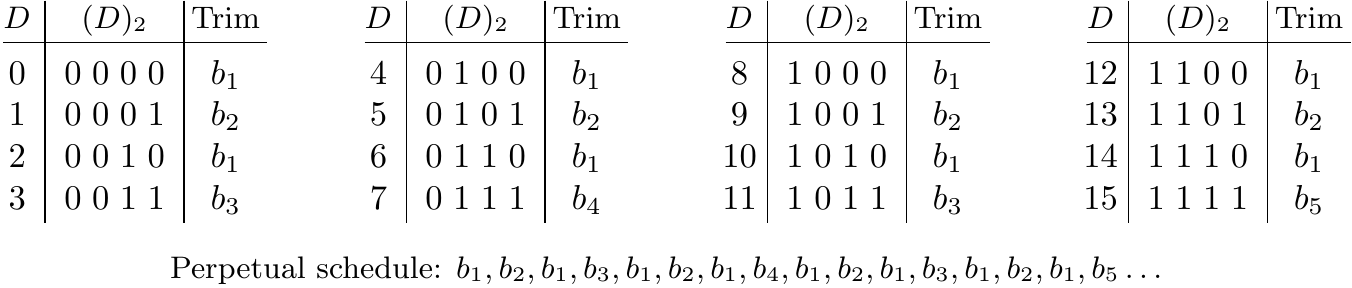}
    \caption{A perpetual schedule of a regular instance with $5$ bamboos.}
    \label{fig:schedule_powers_of_two}
\end{figure}

In this way, the maximum number of days that elapses between any two consecutive cuts of bamboo $b_i$ with $i<n$ is $M_i = 2^i$, while $b_n$ is cut every $M_n = 2^{n-1}$ days. It is then easy to see that, for each bamboo $b_i$, $h_i \cdot M_i = 1$, thus showing that the resulting makespan is $1$ as desired (and this is tight since, in any schedule with bounded makespan, $b_1$ grows for at least $2$ consecutive days). See Figure~\ref{fig:schedule_powers_of_two} for an example with $n=5$.

This above relation immediately suggests the implementation of a Trimming Oracle that maintains the binary representation of $D \bmod 2^{n-1}$. Since it is well known that a binary counter with $n$ bits can be incremented in $O(1)$ amortized time \cite[pp.~454--455]{DBLP:books/daglib/0023376}, we can state the following:
\begin{lemma}
\label{lemma:oracle_regular}
For the special case regular instances, there is a Trimming Oracle that uses $O(n)$ space, can be built in $O(n)$ time, can be queried to report the next bamboo to cut in $O(1)$ amortized time, and achieves makespan $1$.
\end{lemma}

\subsubsection*{A Trimming Oracle for non-regular instances}

Here we show how to design a Trimming Oracle for non-regular instances by iteratively transforming them into suitable regular instances. We will refer to the bamboos $b_1, \dots, b_n$ as \emph{real bamboos} and will introduce the notion of \emph{virtual bamboos}.

A virtual bamboo $v$ represents a collection of (either real or virtual) bamboos whose growth rates yield a regular instance when suitably scaled by a common factor. The growth rate of $v$ will be equal to the sum of the growth rates of the bamboos in its collection.

To see why this concept is useful, consider an example instance $I$ with $6$ bamboos $b_1, \dots, b_6$ with rates $h_1 = \frac{1}{2}$, $h_2 = \frac{1}{8}$, $h_3 = \frac{1}{8}$, $h_4 = \frac{1}{8}$, $h_5 = \frac{1}{16}$, $h_6 = \frac{1}{16}$.
If we replace $h_4$, $h_5$, and $h_6$ with a virtual bamboo $v$ with growth rate $\mathfrak{h}=\frac{1}{8}+\frac{1}{16}+\frac{1}{16} = \frac{1}{4}$ we obtain the related regular instance $I'$ in which the bamboos $b_1$, $v$, $b_2$, $b_3$ have growth rates $\frac{1}{2}$, $\frac{1}{4}$, $\frac{1}{8}$, and $\frac{1}{8}$, respectively.
Notice also that the collection of bamboos associated with $v$ is a regular instance $I_v$ once all the rates are multiplied by $\frac{1}{\mathfrak{h}} = 4$. 
We can now build two Trimming Oracles $\mathcal{O}'$ and $\mathcal{O}_v$ for $I'$ and $I_v$, respectively, by using Lemma~\ref{lemma:oracle_regular}. 
It turns out that $\mathcal{O}'$ and $\mathcal{O}_v$ together allow us to build an oracle $\mathcal{O}_r$ for $I$ as well, which can be represented as a tree (See Figure~\ref{fig:virtual_tree}).
In general, our oracles $\mathcal{O}$ will consist of a tree $T_{\mathcal{O}}$ whose leaves are the real bamboos $b_1, \dots, b_n$ of the input instance and in which each internal vertex $u$ serves two purposes: (i) it represents a virtual bamboo whose associated collection $C$ contains the bamboos associated to the children of $u$; and (ii) it serves as a Trimming Oracle $\mathcal{O}_u$ over the bamboos in $C$.\footnote{The root of $T_{\mathcal{O}}$ can be seen as a virtual bamboo with a growth rate of $1$.}
In order to query $\mathcal{O}$ we proceed as follows: initially we start with a pointer $p$ to the root $r$ of $T_{\mathcal{O}}$; then, we iteratively check whether $p$ points to a leaf $\ell$ or to an internal vertex $u$. In the former case, we trim the real bamboo associated with $\ell$, otherwise we query the Trimming Oracle $\mathcal{O}_u$ associated with $u$ and we move $p$ to the child of $u$ corresponding to the (virtual or real) bamboo returned by the query on $\mathcal{O}_u$. Since all queries on internal vertices can be performed in $O(1)$ amortized time (see Lemma~\ref{lemma:oracle_regular}), the amortized time required to query $\mathcal{O}$ is proportional to the height of $T_{\mathcal{O}}$. See Figure~\ref{fig:virtual_tree} for the schedule associated to our example instance $I$.

\begin{figure}[t]
    \centering
    \includegraphics[width=\textwidth]{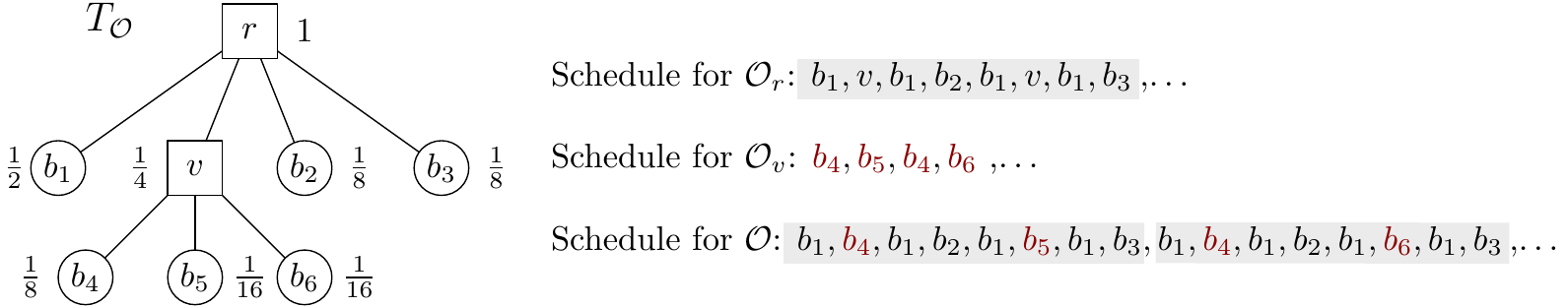}
    \caption{The tree $T_{\mathcal{O}}$ of the Trimming Oracle $\mathcal{O}$ for the instance with $6$ bamboos $b_1,\dots, b_6$ with rates $h_1 = \frac{1}{2}$, $h_2 = \frac{1}{8}$, $h_3 = \frac{1}{8}$, $h_4 = \frac{1}{8}$, $h_5 = \frac{1}{16}$, and $h_6 = \frac{1}{16}$. Bamboos $b_4$, $b_5$, and $b_6$ have been replaced by a virtual bamboo $v$ (and a corresponding oracle $\mathcal{O}_v$) with a virtual growth rate of $\frac{1}{4}$. The root $r$ represents both a virtual bamboo with growing rate $1$ and the corresponding Trimming Oracle $\mathcal{O}_r$ for the regular instance consisting of $b_1$, $v$, $b_2$, and $b_3$.}
    \label{fig:virtual_tree}
\end{figure}

Before showing how to build the tree $T_{\mathcal{O}}$ of our Trimming Oracle $\mathcal{O}$, we prove that the perpetual schedule obtained by querying $\mathcal{O}$ achieves a makespan of at most $1$. At any point in time, we say that the \emph{virtual height} of a virtual bamboo $v$ representing a collection $C$ of (real or virtual) bamboos is the maximum over the (real or virtual) heights of the bamboos in $C$. The bound on the makespan follows by instantiating the following Lemma with $b=r$ and $h=1$, and by noticing that: (i) the root $r$ of $T_{\mathcal{O}}$ is scheduled every day, and (ii) that the maximum virtual height of $r$ is exactly the makespan.

\begin{lemma}\label{lm:oracle_2_apx_correctness}
Let $b$ be a (real or virtual) bamboo with growth rate $h$. If $b$ is scheduled at least once every $\frac{1}{h}$ days, then the maximum (real or virtual) height ever reached by $b$ will be at most $1$.
\end{lemma}
\begin{proof}
The proof is by induction on the number $\eta$ of nodes in the subtree rooted at the vertex representing $b$ in $T_\mathcal{O}$.

If $\eta=1$ then $b$ is a real bamboo and the claim is trivially true since the maximum height reached by $b$ can be at most $h \cdot \frac{1}{h} = 1$.

Suppose then that $\eta \ge 2$ and that the claim holds up to $\eta-1$. Bamboo $b$ must be a virtual bamboo representing some set $C = \{b'_1, b'_2, \dots, b'_k\}$ of (real or virtual)  bamboos which appear as children of $b$ in $T_\mathcal{O}$ and are such that: (i) for $i=1, \dots, k-1$, $b'_i$ has a growth rate of $h'_i = h/2^i$, and (ii) $b'_k$ has a growth rate of $h'_k = h/2^{k-1}$.

Virtual bamboo $b$ schedules the bamboos in $C$ by using the oracle $\mathcal{O}_v$ of Lemma~\ref{lemma:oracle_regular}, on the regular instance obtained by changing the rate of bamboo $b'_i$ from $h'_i$ to $h''_i = h'_i / h$.

Let $d_i$ (resp. $d'_i$) be the maximum number of days between any two consecutive cuts of bamboo $b'_i$ according to to the schedule produced by $\mathcal{O}$ (resp. $\mathcal{O}_v$).
We know that $d'_i \cdot h''_i \le 1$ (as otherwise the schedule of $\mathcal{O}_v$ would result in makespan larger than $1$ on a regular instance, contradicting Lemma~\ref{lemma:oracle_regular}), i.e., $d'_i \le \frac{1}{h''_i}$.
Since, $b$ is scheduled at least every $1/h$ days by hypothesis, we have that $d_i \le \frac{1}{h \cdot h''_i} = \frac{h}{h \cdot h'_i} = \frac{1}{h'_i}$ and hence, by inductive hypothesis, the maximum height reached by $b'_i$ will be at most $1$.
\end{proof}

We now describe an algorithm that constructs a tree $T_\mathcal{O}$ of logarithmic height.

The algorithm employs a collection of sets $S_0, S_1, \dots$, where initially $S_0 = \{ b_1, \dots, b_n \}$ contains all the real bamboos of our input instance, and $S_{i}$ with $i>0$ is obtained from $S_{i-1}$ by performing suitable \merge operations over the bamboos in $S_{i-1}$. 
A \merge operation on a collection $C \subseteq S_{i-1}$ of bamboos, whose growth rates yield a regular instance when multiplied by some common factor, consists of: updating $S_{i-1}$ to $S_{i-1} \setminus C$, \emph{creating} a new virtual bamboo $v$ representing $C$, and adding $v$ to $S_i$.

The algorithm works in phases. At the generic phase $i=1,2,\dots$, it iteratively: (1) looks for a bamboo $b$ with the largest growth rate that can be involved in a \merge operation and (2) perform a \merge operation on a maximal set $C \subseteq S_{i-1}$ among the ones that contain $b$ (and on which a \merge operation can be performed). The procedure is then repeated from step (1) until no suitable bamboo $b$ exists anymore. At this point we name $R_{i-1}$ the current set $S_{i-1}$, we add to $S_i$ all the bamboos in $R_{i-1}$, and we proceed to the next phase. The algorithm terminates whenever the set $S_i$ constructed at the end of a phase contains a single virtual bamboo $r$ (of rate $1$).

The sequence of \merge operations implicitly defines a bottom-up construction of the tree $T_\mathcal{O}$, where every \merge operation creates a new internal vertex associated with its corresponding virtual bamboo. The root of $T_\mathcal{O}$ is $r$ and the height of $T_\mathcal{O}$ coincides with the number of phases of the algorithm. 

\begin{lemma}
\label{lemma:alg_log_n_phases}
The algorithm terminates after at most $O(\log n)$ phases.
\end{lemma}
\begin{proof}
We first prove that the algorithm must eventually terminate.
This is a direct consequence of the fact that, at the beginning of any phase $i$, every set $S_{i-1}$ containing $2$ or more bamboos, admits at least one \merge operation.
Indeed, since \merge operations preserve the sum of the growth rates, the overall sum of the rates of the bamboos in $S_{i-1}$ must be $1$.
Consider now a bamboo $b \in S_{i-1}$ having the lowest growth rate $h$. Since all rates are powers of $\frac{1}{2}$ and must sum to $1$, there must be at least one other bamboo $b' \in S_{i-1} \setminus \{ b \}$ having rate $h$, implying that \merge operation can be performed on $C=\{b, b'\}$.

It remains to bound the number of phases. We prove by induction on $i$ that any internal vertex/virtual bamboo $v$ of $T_\mathcal{O}$ created at phase $i$ has at least $2^i$ leaves as descendants. 
The base case $i=1$ is trivial since the merge operation that created $v$ must have involved at least $2$ real bamboos.

Consider now the case $i \ge 2$. We will show that $v$ was created by a \merge operation on a collection $C$ containing at least $2$ bamboos $v', v''$ that were, in turn, created during phase $i-1$. Hence, by inductive hypothesis, the number of leaves that are descendants of $v$ is the sum of the number of leaves that are descendants of $v'$ and $v''$, respectively, i.e., it is at least $2^{i-1} + 2^{i-1} = 2^i$.

Let $C \subseteq S_{i-1}$ be the set of bamboos used in the \merge operation that created $v$, and let $h$ be the smallest growth rate among the ones of the bamboos in $C$.
Notice that, by definition of \merge operation, there must be $2$ distinct bamboos $v',v''$ with rate $h$ in $C$. We will now show that $v'$ and $v''$ were created during phase $i-1$.
We proceed by contradiction.
If neither of $v'$ and $v''$ were created in phase of $i-1$, then $\{v', v''\} \subseteq R_{i-2}$ which is impossible since $\{v', v''\}$ would have been a feasible merge operation in phase $i-2$.
Assume then that  $v'$ was not created in phase $i-1$, while $v''$ was created in phase $i-1$, w.l.o.g. Then, $v' \in R_{i-2}$, while $v''$ was obtained from a merge operation on a set $C' \subseteq S_{i-2}$ performed in phase $i-1$. Since the growth rate of $v''$ is $h$, the fastest growth rate among the ones of the bamboos in $C'$ must be $h/2$. Hence, the set $C'' = \{v'\} \cup C'$ was a feasible \merge operation in phase $i-1$ when $v''$ was created. This is a contradiction since $C' \subset C''$ was not a maximal set, as required by the algorithm.
 \end{proof}
 
Next Lemma bounds the computational complexity of constructing our oracle.

\begin{lemma}\label{lm:oracle_2apx_buildingtime}
The Trimming Oracle $\mathcal{O}$ can be built in $O(n \log n)$ time. 
\end{lemma}
\begin{proof}
It suffices to prove that every phase $i$ of our algorithm can be implemented in $O(n)$ time, since from Lemma~\ref{lemma:alg_log_n_phases} the number of phases is $O(\log n)$. 

We maintain the set $S_{i-1}$ as a doubly linked list $L_{i-1}$ in which each node $\ell$ is associated with a distinct growth rate $\mathfrak{h}_\ell$ attained by at least one bamboo in $S_{i-1}$ and stores the set $H(\ell)$ of bamboos of $S_{i-1}$ with grow rate $\mathfrak{h}_\ell$. Nodes appear in decreasing order of $\mathfrak{h}_\ell$. The very first list $L_0$ can be constructed in $O(n \log n)$ time by sorting the growth rates of the bamboos in $S_0$. We now show how to build $L_i$ in $O(n)$ time.

The idea is to iteratively find two nodes $\ell_1, \ell_2$ of $L_{i-1}$ such that: (i) $\ell_2$ is not the head of $L_{i-1}$ and appears not earlier than $\ell_1$; (ii) if $\ell_1$ is not the head of $L_{i-1}$, then selecting one bamboo from the set $H(\ell)$ of each node $\ell$ that appears before the predecessor $\ell'_1$ of $\ell_1$, and two bamboos from the set $H(\ell'_1)$ yields the (maximal) set $C$ corresponding the \merge operation that algorithm performs; and (iii) all the bamboos in the sets $H(\ell)$ of the nodes $\ell$ that appear not earlier than $\ell_1$ and before $\ell_2$ in $L_{i-1}$ will not participate in any merge operation of phase $i$. We call the set of these bamboos $D$ (notice that it is possible for $\ell_2$ to be equal to $\ell_1$, in which case no such node $\ell$ exists and $D=\emptyset$).

To find $\ell_1$ and $\ell_2$ notice that $\ell_2$ is the the last node of $L_{i-1}$ for which any two
consecutive nodes preceding $\ell_2$ correspond to consecutive rates\footnote{For technical simplicity, when all consecutive nodes of $L_{i-1}$ correspond to consecutive rates we allow $\ell_1$ and/or $\ell_2$ to point one position past the end of $L_{i-1}$.}, while the predecessor $\ell'_1$ of $\ell_1$ is the last node that appears before $\ell_2$ and such that $|H(\ell'_1)| \ge 2$.

We now delete the bamboos in $C \cup D$ from their respective sets $H(\ell)$ of $L_{i-1}$, create a new virtual bamboo $v$ by a \merge operation on $C$. Finally, delete from $L_{i-1}$ all nodes $\ell$ whose set $H(\ell)$ is now empty. We then repeat this procedure from the beginning until $L_{i-1}$ is empty. 

Concerning the time complexity, notice that finding $\ell_1$ and $\ell_2$ requires $O(k)$ time, where $k$ is the number of nodes that precede $\ell_2$ in $L_{i-1}$. Moreover, all the other steps can be implemented in $O(k)$ time. Therefore, we are able delete $k$ bamboos from $L_{i-1}$ in $O(k)$ time, and hence the overall time complexity to delete all bamboos in $L_{i-1}$ is $O(n)$.

Finally, by keeping track of the sets $D$, of all the virtual bamboos $v$ generated during the iterations, and by using the fact that the rates of the virtual bamboos are monotonically decreasing, it is also possible to build $L_i$ in $O(n)$ time.
\end{proof}

\begin{lemma}\label{lm:oracle_2_apx_space}
The Trimming Oracle $\mathcal{O}$ uses $O(n)$ space.
\end{lemma}
\begin{proof}
    By Lemma~\ref{lemma:oracle_regular} each internal vertex of $T_\mathcal{O}$ maintains a Trimming Oracle with size proportional to the number of its children, implying that the overall space required by $\mathcal{O}$ is proportional to the number $\eta$ of vertices of $T_\mathcal{O}$.
    Since every internal vertex in $T_\mathcal{O}$ has at least $2$ children, we have that $\eta = O(n)$.
\end{proof}

By combing Lemma~\ref{lm:oracle_2_apx_correctness}, Lemma~\ref{lemma:alg_log_n_phases}, Lemma~\ref{lm:oracle_2apx_buildingtime}, and Lemma~\ref{lm:oracle_2_apx_space}, we can state the following theorem that summarizes the result of this section: 

\begin{theorem}
There is a Trimming Oracle that achieves makespan $2$, uses $O(n)$ space, can be built in $O(n \log n)$ time, and can report the next bamboo to trim in $O(\log n)$ amortized time.
\end{theorem}

\subsection*{Acknowledgements}
The authors would like to thank Francesca Marmigi for the picture of the robotic panda gardener in Figure~\ref{fig:example}. We are also grateful to an anonymous reviewer whose comments allowed us to significantly simplify the analysis of \reducemax.

\bibliographystyle{plainurl}
\bibliography{bibliography}

\end{document}